\documentclass[12pt]{amsart}
\usepackage{amsfonts}
\usepackage{amssymb}
\usepackage{textcomp}
\usepackage{graphicx}
\def\sqr#1#2{{\vcenter{\vbox{\hrule height.#2pt
        \hbox{\vrule width.#2pt height#1pt \kern#2pt
        \vrule width.#2pt}
        \hrule height.#2pt}}}}

\newcommand{\nc}{\newcommand}
\nc{\parent}[1]{$[\![#1]\!]$}
\newtheorem{theorem}{Theorem}

\newtheorem{question}{Question}

\newtheorem{remark}{Remark}
\newtheorem{definition}{Definition}

\newtheorem{exercise}{Exercise}
\newenvironment{pf-main}{{\sc Proof of Theorem \ref{mainresult}.}\hspace{3mm}}{\qed}
\nc{\esssup}{\mathrm{ess}\mbox{ }\mathrm{sup}}
\nc{\cadlag}{c\`{a}dl\`{a}g } \nc{\ba}{\begin{array}} \nc{\ea}{\end{array}} \nc{\be}{\begin{equation}}
\nc{\ee}{\end{equation}} \nc{\bea}{\begin{eqnarray}} \nc{\eea}{\end{eqnarray}} \nc{\bean}{\begin{eqnarray*}}
\nc{\eean}{\end{eqnarray*}} \nc{\bu}{\bullet} \nc{\nn}{\nonumber} \nc{\cA}{{\mathcal A}} \nc{\cB}{{\mathcal B}}
\nc{\cC}{{\mathcal C}} \nc{\cD}{{\mathcal D}} \nc{\bbD}{\mathbb{D}} \nc{\cG}{{\mathcal G}} \nc{\cF}{{\mathcal F}} \nc{\cN}{{\mathcal N}} 
\nc{\cS}{{\mathcal S}} \nc{\cU}{{\mathcal U}} \nc{\cH}{{\mathcal H}} \nc{\cJ}{{\mathcal J}}\nc{\cK}{{\mathcal K}} \nc{\cM}{{\mathcal M}}
\nc{\cP}{{\mathcal P}} \nc{\bbE}{\mathbb{E}} \nc{\bbEQ}{\mathbb{E}_{\mathbb{Q}}} \nc{\eps}{\varepsilon}
\nc{\bbEP}{\mathbb{E}_{\mathbb{P}}}\nc{\bbL}{\mathbb{L}}  \nc{\bbH}{\mathbb{H}} \nc{\bbG}{\mathbb{G}}
\nc{\bbF}{\mathbb{F}}\nc{\bbP}{\mathbb{P}} \nc{\bbQ}{\mathbb{Q}}
\nc{\Om}{\Omega} \nc{\om}{\omega} \nc{\bbR}{\mathbb{R}} \nc{\bbC}{\mathbb{C}} \nc{\bfr}{\begin{flushright}}
\nc{\efr}{\end{flushright}} \nc{\dXt}{\Delta X_{t}} \nc{\dXs}{\Delta X_{s}} \nc{\bs}{\blacksquare} \nc{\dX}{\Delta
X} \nc{\dY}{\Delta Y} \nc{\half}{\frac{1}{2}} \nc{\dnkx}{\left(X(T^{n}_{k})-X(T^{n}_{k-1})\right)} \nc{\zmax}{Z^{\max}}
\nc{\zmin}{Z^{\min}}

\def\rar{\rightarrow}

\nc{\chf}{\mbox{$\mathbf1$}}
 \nc{\eid}{\stackrel{d}{=}}
\nc{\what}{\widehat} \nc{\fhat}{\what{f}}\nc {\parx}{\frac{\partial}{\partial x}} \nc
{\parw}{\frac{\partial}{\partial w}} \nc
{\parww}{\frac{\partial^2}{\partial w^2}} \nc{\bfT}{\mathbf{T}}
\begin{document}
\title{Mathematics of Market Microstructure under Asymmetric Information}
	\author{Umut \c{C}et\i n}\thanks{Lecture notes for the summer course given for 2018 Mathematical Finance Summer School at Shandong Unversity.}
\address{London School of Economics, Department of Statistics, Columbia House, Houghton Street, London WC2A 2AE}\date{\today}
\maketitle
\section{Motivation}
A vast number of financial assets change hands every day and, as a
result, the prices of these assets fluctuate continuously. What drives the asset prices is the expectation of the assets' future
payoffs and a price is set when a buyer and a seller agree on an
exchange. At any given point in time, usually there are numerous
agents in the market who are interested in trading. Some agents are
individuals with relatively small portfolios while others are
investments banks or hedge funds who are acting on behalf of large
corporations or a collection of individuals. Obviously, these agents
have different attitudes toward risk and do not have equal access to
trading technologies, information or other relevant resources. Moreover, it is very rare that the
cumulative demand of the buyers will be  met by the total shares offered by the
sellers without any excess supply. An imbalance of demand and
supply is the rule rather than an exception in today's markets, which brings {\em liquidity risk} to the fore. These
features of the modern markets challenge the conventional asset pricing
theories which assume {\em perfect competition} and {\em no} liquidity risk.

Market microstructure theory provides an alternative to {\em
	frictionless} models of trading behaviour that assume perfect
competition and free entry. To quote O'Hara \cite{MOH}, ``[It] is the
study of the process and outcomes of exchanging assets under explicit
trading rules.'' Thus, market microstructure   analyses how a specific trading mechanism or heterogeneity of traders affect the price formation,
comes up with measures for market liquidity and studies the
sensitivity of liquidity and other indicators of market behaviour on
different trading mechanisms and information heterogeneity. 

The primary goal of Market Microstructure (MS) models is to understand the {\em temporary} and
{\em permanent} impacts of the trades on the asset price and how the
price-setting rules evolve in time.  In real markets bid
and ask prices are announced  by {\em specialists} or {\em dealers},
whom we will collectively call {\em market makers} henceforth.  The early literature on market
microstructure (\cite{Garman}, \cite{S78}, \cite{AM80}, and \cite{HS81})  have started with the simple observation that
the trades could involve some
implicit costs due to the need for immediate execution, which is
provided by the market makers. At the same
time, the market makers take into account their inventory level when
making pricing decisions.  These works have concluded that the market makers adjust the prices in
order to keep their inventories around a certain level in the long
run: they lower the price when their inventory levels are too high and raise
the prices when they are short large quantities. As the market makers
want to keep their inventories around a fixed level, the impact of
trades are {\em transitory} since the prices are also expected to mean
revert.

The MS research have later shifted its focus to models with asymmetric
information, which account for permanent changes in the price. 
The canonical model of markets with asymmetric information is due to
Kyle \cite{K}. Kyle studies a market for a single risky asset whose
price is determined in equilibrium in discrete time. The key feature of this model is that the market makers cannot distinguish between the
informed and uninformed trades and compete to fill the net demand.  In this model market makers `learn’
from the net demand by `filtering’ what the informed trader knows,
which is `corrupted' by the demand of the uninformed traders. The market makers learn from the order flow and they update their pricing strategies as a result of this learning mechanism.

To get a flavour of the Kyle model suppose that there is an asset whose value $V$ will be revealed at time $1$. Assume further the existence of an insider who knows the value of $V$ at time $0$.  To simplify the matters the insider will be allowed to trade once at time $0$ and liquidate her position at time $1$. 

At time $0$ there are also {\em noise traders} who are not strategic and their cumulative demand for the asset is given by $\nu \sim N(0,\sigma_n^2)$. Consistent with the term 'noise' $\nu$ is assumed to be independent of $V$. 

If the insider trades $\theta$ many shares, the market makers observe the net demand $Y:=\theta+n$ and take the opposite side to clear the market by setting a price. They know the distribution of $V$ but no other relevant information regarding its value. The market makers are {\em risk neutral} and compete in a Bertrand fashion to fill the aggregate order $Y$. That is, the price $h(y)$ chosen by the market makers for $Y=y$ is such that their expected profit is $0$. Since they will also liquidate their position at time $1$ at price $V$, this implies
\be \label{e:mart}
h(y)=E[V|Y=y].
\ee
Given this pricing rule $y\mapsto h(y)$ the insider needs to find her optimal trading amount. Suppose further that $V\sim N(\mu, \sigma^2)$. We are interested in a Nash-type equilibrium: $(\theta,h)$ will constitute and equilibrium if
\begin{enumerate}
	\item Given $h$, $\theta$ maximises the expected profit of the insider;
	\item Given $\theta$, $h$ satsifies (\ref{e:mart}).
\end{enumerate} 
Let us now prove that a {\em linear equilibrium} in which $h(y)=a + by$ and $\theta=\alpha + \beta V$ exists.  First observe that if $h(y)=a+\lambda y$, the insider's optimisation problem given $V=v$ is
\[
\max_{\alpha, \beta } E[(\alpha +\beta v)(v- a-\lambda(\nu+ \alpha +\beta v)).
\]
The profit/loss is quadratic in parameters and the first order condition yields:
\be \label{e:FOC}
\alpha +\beta v= \frac{v-a}{2\lambda}.
\ee
On the other hand, (\ref{e:mart}) requires
\[
a+ \lambda Y= E[V| Y].
\]
Now, since $(V,\nu)$ is a Gaussian vector, the conditional distribution of $V$ given $Y$ is also Gaussian, which can be determined by Bayes' rule. 
Formally,
\[
P(V\in dv |Y=y)= P(Y\in dy|V=v)\frac{P(V\in dv)}{P(Y\in dy)}\sim P(Y\in dy|V=v)P(V\in dv).
\]
Moreover, given $V=v$, $Y=\nu+ \alpha + \beta V\sim N(\alpha+\beta v, \sigma_{\nu}^2)$. Thus, $P(Y\in dy|V=v)$ is proportional to
\[
\exp\left(-\frac{(y-\alpha-\beta v)^2}{2\sigma_{\nu}^2}\right).
\]
Hence,
\[
P(V\in dv |Y=y)\sim\exp\left(-\frac{(y-\alpha-\beta v)^2}{2\sigma_{\nu}^2}-\frac{v^2}{2\sigma^2}\right)\sim\exp(-\frac{(v-\hat{\mu})^2}{2\Sigma^2}),
\]
where
\[
\frac{1}{\Sigma^2}=\frac{1}{\sigma^2}+\frac{\beta^2}{\sigma_{\nu}^2},\quad \hat{\mu}=\beta(y-\alpha)\frac{\Sigma^2}{\sigma^2_Z}.
\]
That is, $V$ is Gaussian with mean $\hat{\mu}$ and variance $\Sigma^2$ given $Y=y$. Thus,
\[
a+\lambda y=\beta(y-\alpha)\frac{\Sigma^2}{\sigma^2_Z}=\beta(y-\alpha)\frac{\sigma^2}{\beta^2\sigma^2+\sigma_{\nu}^2},
\]
which in turn yields
\[
\lambda=\frac{\beta \sigma^2}{\beta^2\sigma^2+\sigma_{\nu}^2} \mbox{ and } a=- \alpha\lambda.
\]
 Recall that (\ref{e:FOC}) implies $2\lambda \beta =1$. Thus,
 \[
 \beta \sigma^2 =\frac{\beta\sigma^2}{2}+\frac{\sigma^2_Z}{2\beta}.
 \]
 Consequently, $\beta =\frac{\sigma_{\nu}}{\sigma}$ and $\lambda=\frac{\sigma}{2\sigma_{\nu}}$.
The remaining two equations for $a$ and $\alpha$ are satisfied only if $a=\alpha=0$. 
\begin{question}
	What is the interpretation of equilibrium values of $\lambda$ and $\beta$?
\end{question}
\begin{question}
	Without making all the calculations above can you guess the equilibrium value of $a$ and $\alpha$ if the mean of $V$ were different than $0$?
\end{question}
{\em Value of information:} Given the above explicit characterisation of equilibrium we can compute the equilibrium level of wealth of the insider, which is given by
\[
(1-\lambda\beta)\beta v^2=\beta \frac{v^2}{2}
\]
Thus, the {\em ex-ante}, i.e. unconditional, value of information equals
\be \label{e:infval}
\beta \frac{\sigma^2}{2}=\frac{\sigma\sigma_{\nu}}{2}.
\ee
\begin{question}
	What is the interpretation of (\ref{e:infval})?
\end{question}
\section{Kyle model in continuous time}
If a trader has some private information regarding the future value of the asset, she would like to take advantage of this and trade dynamically, not just once as above. The continuous time version of the Kyle model is formalised by K. Back \cite{B}. Although in the literature it is usually assumed that the informed investor knows the future asset value perfectly, this is not a necessary assumption as we shall soon see. 

Let us suppose that the time-1 value of the traded asset is given by some random variable $V$, which will become public knowledge at $t=1$ to all market participants.  

We shall work on a filtered probability space $(\Omega, \cG, (\cG_t)_{t\in [0,1]}, \bbQ)$.  

Three types of agents  trade in the market. They differ in their information sets, and objectives, as follows.

\begin{itemize}
	\item \textit{Noise/liquidity traders} trade for liquidity reasons, and
	their total demand at time $t$ is given by a standard $(\cG_t)$-Brownian motion $B$ independent of $Z$. (We normalise the variance of the noise trades so that it is given by a Brownian motion with unit variance)
	\item \textit{Market makers} observe only the total demand
	\[
	Y=\theta+B, 
	\]
	where $\theta$ is the demand process of the informed trader. The admissibility condition imposed later on $\theta$ will entail in particular that $Y$ is a semimartingale. 
	
	They set the price of the risky asset via a {\em Bertrand competition} and clear the market. We assume that the market makers set the price as a function of the total order process at time $t$, i.e. we
	consider pricing functionals $S\left( Y_{[0,t]},t\right) $ of the following form%
	\begin{equation} \label{mm:e:rule_mm}
	S\left( Y_{[0,t]},t\right) =H\left(t, Y_t\right), \qquad \forall t\in [0,1).
	\end{equation}
	 Moreover, a pricing rule $H$ has to be admissible in the sense of Definition \ref{mm:d:prule}.  In particular, $H \in C^{1,2}$ and, therefore, $S$ will be a semimartingale as well.
	\item \textit{The informed trader (insider)} observes the price process $
	S_{t}=H\left(t, Y_t\right)$  and her private signal, $Z$. Based in her signal, she makes an educated guess about $V$. Thus, there exists a measurable function $f$ such that
	\[
	f(Z)=E[V|Z].
	\]
	We assume that $Z$ is a continuous random variable and $f$ is continuous. Thus, without loss of generality we can take $Z$ to be a standard Normal random variable - possibly the time-1 value of some Brownian motion. Moreover, $f$ can be taken strictly increasing (why?). This entails in particular that  the larger the signal $Z$ the larger the value of the risky asset for the informed trader. 
	
	She is assumed to be 
	risk-neutral, her objective is to maximize the expected final wealth.  
	\bean
	&&\sup_{\theta \in \mathcal{A}(H)}E^{0,z}\left[ W_{1}^{\theta
	}\right], \mbox{ where} 
	\\
	W_{1}^{\theta
	}&=& (V-S_{1-}))\theta_{1-}+ \int_0^{1-}\theta_{s-}dS_s.
\eean
However, using the tower property of conditional expectations, the above problem is equivalent to
\bea
&&\sup_{\theta \in \mathcal{A}(H)}E^{0,z}\left[ W_{1}^{\theta
}\right], \mbox{ where} \label{ins_obj}\\
W_{1}^{\theta
}&=&	(f(Z)-S_{1-})\theta _{1-}+\int_{0}^{1-}\theta _{s-}dS_{s}. \label{mm:eq:insW}
	\eea
	In above $\mathcal{A}(H)$ is the set of admissible trading strategies
	for the given pricing rule\footnote{Note that this implies  the insider's
		optimal trading strategy takes into account the \emph{feedback
			effect}, i.e. that prices react to her trading strategy.} $H$, which will be defined in Definition \ref{mm:d:iadm}. Moreover, $E^{0,z}$ is the expectation with respect to $P^{0,z}$, which is the probability measure on $\sigma(Y_s, Z; s\leq 1)$ generated by  $(Y,Z)$ with $Y_0=0$ and $Z=z$.  
	
	Thus, the insider maximises the
	expected value of her final wealth
	$W_{1}^{\theta }$, where the first term on the right hand side of equation (%
	\ref{ins_obj}) is the contribution to the final wealth due to a potential
	differential between  the market price and the fundamental value or insider's private valuation at the time of information
	release, and the second term is the contribution to the final wealth coming from
	the trading activity.
\end{itemize}

Given the above market structure we can now precisely define the filtrations of the market makers and of the informed trader.  As we shall require them to satisfy the usual conditions, we first define the probability measures that will be used in the completion of their filtrations. 

First  define $\cF:=\sigma(B_t,Z; t \leq 1)$ and let $P^{0,z}$ be the probability measure on $\cF$ generated by $(B,Z)$ with $B_0=0$ and $Z=z$.  Also define the probability measure $\bbP$ on $(\Om, \cF)$  by
\be \label{mm:d:bbP}
\bbP(E)=\int_{\bbR} Q^{0,z}(E) \bbQ(Z_0\in dz),
\ee
for any $E \in \cF$.  
\subsection{Some technicalities regarding filtrations and null sets}
While $P^{0,z}$ is how the informed trader assign likelihood to the events generated by $B$ and $Z$, $\bbP$ is the probability distribution of the market makers who do not observe $Z_0$ exactly. Thus, the market makers' filtration, denoted by $\cF^M$, will be the right-continuous augmentation with the $\bbP$-null sets of the filtration generated by $Y$.

On the other hand, since the informed trader knows the value of $Z_0$ perfectly, it is plausible to assume that her filtration is augmented with the $P^{0,z}$-null sets. However, this will make the modelling cumbersome since the filtration will have an extra dependence on the value of $Z_0$ purely for technical reasons.  Another natural choice is to consider the intersection of all these augmentations. That is,
\[
\cF^I_t=\cap_{z\in \bbR}\cF_t^{I,z},
\]
where $\cF_t^{I,z}$ is the usual augmentation of the filtration generated by $Z$ and $S$ and completed with the null sets of $P^{0,z}$.

{\em \underline{Key observation for the filtrations:}} The filtration of the insider is the {\em enlargement} of the market filtration with her own signal $Z$.

{\em \underline{Key observation for the probability measures:}} The probability measure of the insider is {\em singular} with respect to the probability measure of the market makers. Indeed, $P^{0,z}(F(Z)=f(z))=1$ while $\bbP(f(Z)=f(z))=0$. 
\subsection{Definition of equilibrium}
We are finally in a position to give a rigorous definition of  the rational expectations equilibrium of this market, i.e. a pair consisting of
an \emph{admissible} pricing rule and an \emph{admissible}
trading strategy such that: \textit{a)} given the pricing rule
the trading strategy is optimal, \textit{b)} given the trading
strategy,  the pricing rule is {\em rational} in the following sense:
\be \label{mm:d:mm_obj}
H(t,Y_t)=S_t=\mathbb{E}\left[f(Z_1)|\mathcal{F}_t^M\right],
\ee
where $\bbE$ corresponds to the expectation operator under $\bbP$.
To formalize
this definition of equilibrium, we first  define the sets of admissible
pricing rules and trading strategies.

\begin{definition}\label{mm:d:prule} An {\em admissible
		pricing rule} is any function $H$ fulfilling the following
	conditions:
	\begin{enumerate}
		\item $H \in C^{1,2}([0,1) \times \bbR)$.
		\item $x \mapsto H (t,x)$ is strictly increasing for every $t\in [0,1)$;
	\end{enumerate}
\end{definition}
\begin{remark} \label{mm:r:insfilt} The strict monotonicity of $H$ in the space variable implies $H$ is invertible prior to time $1$,
	thus, the filtration of the insider is generated by $Y$ and $Z$. This in turn implies that $(\cF^{S,Z}_t)=(\cF^{B,Z}_t)$, i.e. the insider has full information about the market. 
\end{remark}
In view of the above one can take  $\cF^I_t=\cF^{B,Z}_t$ for all $t\in [0,1]$. 
\begin{definition} \label{mm:d:iadm}
	An $\cF^{B,Z}$-adapted  $\theta$ is said to be an  admissible trading
	strategy for a  given pricing rule $H$  if
	\begin{enumerate} \item $\theta$ is  a semimartingale on $(\Om, \cF, (\cF^{B,Z}_t), Q^{0,z})$ with summable jumps for each $z \in \bbR$,
		\item and no doubling strategies are
		allowed, i.e. for all $z \in \bbR$
		\begin{equation}
		E^{0,z}\left[ \int_{0}^{1}H^{2}\left(t,X_t\right)dt\right] <\infty.
		\label{mm:e:theta_cond_2}
		\end{equation}
		The set of admissible trading strategies for the  given  $H$ is denoted with $\mathcal{A}(H)$.
	\end{enumerate}
\end{definition} 

It is standard (see, e.g., \cite{BP}  or \cite{Wu}) in the insider trading literature to limit
the set of admissible strategies to absolutely continuous ones motivated by the result in Back \cite{B}. We will prove  this in Theorem \ref{mm:t:AC} that the insider does not make any extra gain if she does not employ continuous strategies of finite variation. 

In view of the above discussion we can limit the admissible strategies of the insider to the absolutely continuous ones denoted by $\cA_c (H)$. Indeed, Theorem \ref{mm:t:AC} shows that under a mild condition the value function of the insider is unchanged if the insider is only allowed to use absolutely continuous strategies. Moreover, even if these conditions are not met, Theorem \ref{mm:t:AC} also demonstrates that an absolutely continuous strategy that brings the market price to the fundamental value of the asset is optimal within $\cA$.  In the models that we consider in the following chapters such a strategy will always exist and the equilibrium pricing rule will satisfy the conditions of Theorem \ref{mm:t:AC}. Consequently the restriction to $\cA_c$ is without loss of generality. 

Now we can formally define the market equilibrium as follows.
\begin{definition} \label{eqd} A couple $(H^{\ast} \theta^{\ast})$
	is said to form an equilibrium if $H^{\ast}$ is an admissible pricing rule,
	$\theta^{\ast} \in \cA_c(H^{\ast})$, and the following conditions are
	satisfied:
	\begin{enumerate}
		\item {\em Market efficiency condition:} given $\theta^{\ast}$,
		$H^{\ast}$ is a rational pricing rule, i.e. it satisfies (\ref{mm:d:mm_obj}).
		\item {\em Insider
			optimality condition:} given $H^{\ast}$, $\theta^{\ast}$ solves
		the insider optimization problem:
		\[
		\bbE[W^{\theta^{\ast}}_1] = \sup_{\theta \in \cA_c(H^{\ast})} \bbE [W^{\theta}_1].
		\]
	\end{enumerate}
\end{definition}
\section{On insider's optimal strategy}

Before showing that the strategies with discontinuous paths or with paths of infinite variation are suboptimal let us informally deduce the Hamilton-Jacobi-Bellmann (HJB) equation associated to the value function of the insider assuming absolutely continuous trading strategies. 

 Let $H$ be any rational pricing rule and suppose that $d\theta_t=\alpha_t dt$. First, notice that a standard application of
integration-by-parts formula applied to $W_1 ^\theta$  gives
\be \label{eq:w2} W_1 ^\theta = \int_0^1 (f(Z_1)-S_s) \alpha_s \,ds .\ee Furthermore,
\be \label{eq:cew2}
E^{0,z}\left[\int_0 ^1 (f(Z_1)-S_s) \alpha_s ds \right]=E^{0,z}\left[\int_0 ^1 (f(z)-S_s) \alpha_s ds \right].
\ee

In view of (\ref{eq:w2}) and (\ref{eq:cew2}), insider's optimization problem becomes
\begin{equation}\label{OptInsider}
\sup_{\theta } E^{0,z} [ W_1 ^\theta ] = \sup_{\theta }
E^{0,z}\left[\int_0 ^1 (f(z)- H(s, Y_s)) \alpha_s ds
\right].\end{equation}

Let us now introduce the value function of the insider:
\[
\phi(t,y,z) := \esssup_{\alpha} E^{0,z} \left[ \int_t ^1
(f(z)-H(s,Y_s))\alpha_sds | Y_t =y, Z = z \right], \quad t\in
[0,1]. \] 
Applying formally the dynamic programming principle, we get the following HJB equation:
\begin{equation}
0=\sup_{\alpha}\left(\left[ \phi_y
+f(z)-H(t,y)\right]\alpha \right)+\phi_t+\frac{1}{2} 
\phi_{yy}.
\end{equation}
Thus, for the finiteness of the value function and the existence of an optimal $\alpha$ we need
\begin{eqnarray}\label{mm:eq:wf_1}
\phi_y+f(z)-H(t,y)=0\\ \label{mm:eq:wf_2} \phi_t+\frac{1}{2}\phi_{yy}=0.
\end{eqnarray}
Differentiating (\ref{mm:eq:wf_1}) with respect to $y$
and since from (\ref{mm:eq:wf_1}) it follows that
$\phi_y=H(t,y)-f(z)$, we get 
\be \label{mm:eq:wf_4}
 \phi_{yy}=H_y(t,y), \; \phi_{yyy}=H_{yy}. \ee
 	 Since differentiation
(\ref{mm:eq:wf_1}) with respect to $t$ gives
$$\phi_{yt}=H_t(t,y),$$
(\ref{mm:eq:wf_4}) implies after differentiating (\ref{mm:eq:wf_2}) with respect to $y$ \be 
H_t(t,y)+\half H_{yy}(t,y)=0. \label{mm:e:pdeh} \ee

Thus, the last two equations seem to be necessary to have a finite solution to the insider's problem. 

The next result shows that the equation (\ref{mm:e:pdeh}) also implies the insider must use continuous strategies of finite variation. 
\begin{theorem} \label{mm:t:AC} Let  $H$ be an admissible pricing rule satisfying (\ref{mm:e:pdeh}).
	Then  $\theta \in \cA(H)$ is an optimal strategy if 
	\begin{enumerate}
		\item[i)] $\theta$ is continuous and of finite variation,
		\item[ii)]   and $H(1-,Y_{1-})=f(Z), \, P^{0,z}$-a.s..
	\end{enumerate}

	Moreover, if we further assume that $H $ is bounded, 
	then  there exists an admissible {\em absolutely continuous} strategy $\theta$ such that
	\[
	\sup_{\theta \in \mathcal{A}(H)}E^{0,z}\left[W_1^{\theta}\right]=E^{0,z}\left[W_1^{\theta}\right].
	\]
\end{theorem}
The proof of this theorem and construction of the equilibrium in the Kyle model will be given after we collect some machinery. 
As seen from the statement the insider must use a {\em bridge strategy} when trading optimally. This requires certain results  on the {\em conditioning of diffusion processes}. While doing so we shall also pay attention to its connection with the {\em theory of enlargement of filtrations} since the insider's filtration is the enlargement of the filtration of market makers. Also note that the rationality of market makers' pricing rule means they  need to compute the conditional distribution of $Z$ given the history of $Y$. Thus, we will need to review the basics of {\em stochastic filtering}, at least in the context of Gaussian processes. 
\section{Stochastic filtering}
Suppose that we are given on a filtered probability space  an
adapted process of interest, $X=(X_t)_{0 \leq t \leq T}$, called the
{\em signal process}, for a deterministic $T$. The problem is that
the signal cannot be observed directly and all we can see is an
adapted {\em observation process} $Y=(Y_t)_{0 \leq t \leq T}$. The
filtering is concerned with finding $\bbE[f(X_t)|\cF^Y_t]$, where
$\cF^Y$ is the minimal filtration generated by $Y$ and satisfying
the usual hypotheses, and $f$ is a measurable function.
\begin{remark} There is a problem with the definition of the process
	$(\bbE[f(X_t)|\cF^Y_t])_{0 \leq t \leq T}$ as the conditional
	expectation $\bbE[f(X_t)|\cF^Y_t]$ is only defined $a.s.$ and there
	are uncountably many $t$ between $0$ and $T$! However, there exists
	a process, let's denote it with $f^o$, called the {\em
		$\cF^Y$-optional projection} of $f(X)$, which satisfies
	$f^o_t=\bbE[f(X_t)|\cF^Y_t]$, for every $t$ (and some more).
	Moreover, $f^o$ is uniquely defined. Thus, whenever we define a
	process by $(\bbE[H_t|\cF^Y_t])_{0 \leq t \leq T}$, we shall always
	mean, and use, the $\cF^Y$-optional projection of $H$. The
	$F^Y$-optional projection of $H$ will be denoted with $\widehat{H}$.
	See the second volume of Rogers and Williams for more details. We
	also suppose the filtration supports two Brownian motions, $W$ and
	$B$, such that $d[W,B]_t=\rho_t dt$, for some predictable process
	$\rho$.
\end{remark}
\subsection{The innovations approach to nonlinear filtering}
Let's suppose the observation process is of the form \be \label{why}
Y_t=\int_0^t h_s ds + W_t, \ee where $W$ is a standard Brownian
motion and $h$ is an adapted process such that \be
\label{L2h} \bbE\left(\int_0^T h_s^2 ds\right) < \infty. \ee The
main result of nonlinear filtering theory is the following:
\begin{theorem} \label{inn} {\bf (Fujisaki, Kallianpur and Kunita)}
	\begin{enumerate}
		\item The process $N$ defined by
		\be \label{defN} N_t=Y_t - \int_0^t \widehat{h}_sds, \ee for each $t
		\in [0,T]$, is an $\cF^Y$-Brownian motion.
		\item If $M$ is an $L^2$-bounded $\cF^Y$-martingale with $M_0=0$,
		then there exists an $\cF^Y$-predictable process $C$ such that
		\[
		\bbE\left(\int_0^T C_s^2 ds\right) < \infty,
		\]
		and that
		\[
		Z_t = \int_0^t C_s dN_s.
		\]
	\end{enumerate}
\end{theorem}
The $F^Y$-Brownian motion $N$ is called the {\em innovation process}
in filtering literature.

\begin{proof} Let $S$ be an $\cF^Y$ stopping time. Since we only
	observe $Y$ over the finite interval $[0,T]$, $S \leq T$, hence
	bounded. Let $N^{\ast}_T=\sup_{t \leq T}|N_t|$. Note that
	$N^{\ast}_T$ is dominated by the random variable
	\[
	W^{\ast}_T + \int_0^T\left\{|h_s|+|\widehat{h}_s|\right\}ds,
	\]
	which is square integrable due to (\ref{L2h}). Thus, $N_S$ is also
	integrable and \bean \bbE(N_S)&=&\bbE\left\{W_S +
	\int_0^S(h_s-\widehat{h}_s)ds\right\} \\
	&=&\int_0^T\bbE\left[(h_s-\widehat{h}_s)1 _{[s \leq S]}\right]ds=0,
	\eean where we used the optional stopping theorem for $\bbE[W_S]$ in
	order to get the first equality, the fact that
	$\widehat{h}_s=\bbE[h_s|\cF^Y_s]$ and $[s \leq S] \in \cF^Y_s$ to
	arrive at the last equality. This shows $N$ is an
	$\cF^Y$-martingale. Since $[N,N]_t=t$, for every $t \in [0,T]$, this
	shows $N$ is an $\cF^Y$-Brownian motion by L\'evy's
	characterisation. See Rogers and Williams \cite{RW2} Chapter VI.8 for the proof
	of the second part.
\end{proof}

Let the signal process X have the following differential: \be
\label{signal} dX_t=\alpha_t dt + \eta_t dB_t, \ee where $\alpha$ is
adapted and $\eta$ is predictable. We further suppose
\[
\bbE\left(\int_0^T \alpha_s^2 ds\right) < \infty,
\]
and
\[
\sup_{t \in [0,T]} \bbE X_t^2 <\infty.
\]

\begin{theorem} \label{filter} Let $X$ and $Y$ be as above. Then we have the
	following fitering equations:
	\[
	\widehat{X}_t= \widehat{X}_0+ \int_0^t \widehat{\alpha}_s ds + \int_0^t
	\left\{\widehat{X_s h_s} -\widehat{X}_s\widehat{h}_s + \widehat{\eta_s
		\rho_s}\right\}dN_s.
	\]
\end{theorem}
\begin{proof}
	First notice that if $C$ is an adapted process such that
	\[
	\bbE \int_0^T |C_s|ds < \infty,
	\]
	and if $V_t= \int_0^t C_s ds$, then \be \label{mg} \what{V}_t -
	\int_0^t \what{C}_s ds \mbox{ is an $\cF^Y$-martingale.} \ee In
	order to prove this it suffices to prove for any $\cF^Y$-stopping
	time $S \leq T$, $\bbE\what{V}_S = \bbE \int_0^S \what{C}_s ds.$
	Indeed, \bean
	\bbE\what{V}_S=\bbE V_S&=&\bbE \int_0^S C_s ds \\
	&=&\int_0^T \bbE\left[1_{[s \leq S]}C_s\right] ds \\
	&=&\int_0^T \bbE\left[1_{[s \leq S]}\what{C}_s\right] ds\\
	&=&\bbE \int_0^S \what{C}_s ds . \eean This in turn implies
	\[
	M_t:=\what{X}_t - \what{X}_0 - \int_0^t \what{\alpha}_s ds,
	\]
	is an $\cF^Y$ martingale with $M_0=0$. Moreover, it is an
	$L^2$-bounded martingale due to the assumed integrability conditions
	on $\alpha$ and $X$. Thus, by Theorem \ref{inn} there exists a
	predictable process $\phi$ such that
	\[
	M_t=\int_0^t \phi_s dN_s.
	\]
	Next, we shall calculate the martingale $M$ explicitly. In order to
	do this we will calculate the optional projection of $XY$ in two
	different ways.  Using integration by parts \bean X_t Y_t&=&
	\int_0^t X_s dY_s + \int_0^t Y_s dX_s +[X,Y]_t \\
	&=&\int_0^t \left\{X_s h_s + Y_s \alpha_s+ \eta_s\rho_s\right\}ds +
	\mbox{ $\cF$-martingale}. \eean Therefore, using (\ref{mg}) \be
	\label{fe1} \what{X_t Y_t}=\what{X_t}Y_t=\int_0^t \left\{\what{X_s
		h_s} + Y_s \what{\alpha_s}+ \what{\eta_s\rho_s}\right\}ds + \mbox{
		$\cF^Y$-martingale}. \ee
	The other way is the following:
	\bea
	\what{X}_t Y_t &=& \int_0^t \what{X}_t dY_s + \int_0^t Y_s
	d\what{X}_s +
	[\hat{X},Y]_t  \nn \\
	&=&\int_0^t \what{X}_t \left\{dN_s +\what{h}_sds\right\} +\int_0^t
	Y_s\left\{dM_s + \what{\alpha}_sds\right\} + [M,N]_t \nn
	\\  &=&\int_0^t \left\{\what{X}_s \what{h}_s + Y_s \what{\alpha}_s +
	\phi_s\right\}ds  + \mbox{ $\cF^Y$-martingale} \label{fe2} \eea
	
	(\ref{fe1}) and (\ref{fe2}) together imply
	\[
	\int_0^t \left\{\what{X}_s \what{h}_s + Y_s \what{\alpha}_s +
	\phi_s\right\}ds - \int_0^t \left\{\what{X_s h_s} + Y_s
	\what{\alpha_s}+ \what{\eta_s\rho_s}\right\}ds
	\]
	is an $\cF^Y$-martingale, thus, must be zero being, of finite
	variation. This implies
	\[
	\phi_s = \what{X_s h_s} - \what{X}_s \what{h}_s +
	\what{\eta_s\rho_s},
	\]
	for each $s$. This proves the desired filtering equation.
\end{proof}
\subsection{The Markov case}
Observe that we have not made any Markov assumption on $X$ or $Y$. We
will now take a look at this special case and obtain equations that
determine the conditional distribution of $X$.

Let's suppose $X$ is a diffusion with generator $L$:
\[
L=\frac{1}{2}\sigma^2(x)\frac{d^2}{dx^2} +b(x) \frac{d}{dx}.
\]
We will also suppose that $B$ and $W$ are independent for simplicity
and  $h_s=h(X_s)$ and
\[
\bbE\left[\int_0^T h_s^2\,ds\right]<\infty.
\]
Then, using the already obtained formulas we obtain the following
\begin{theorem} Let $f \in C^2_K$ and define $\pi_t
	f=\bbE[f(X_t)|\cF^Y_t]$. Then,
	\[
	\pi_tf =\pi_0f + \int_0^t \pi_s Lf \, ds +\int_0^t \{\pi_s hf
	-\pi_sh \pi_s f\}\,dN_s.
	\]
\end{theorem}
The equation in the above theorem is called {\em Kushner-Stratonovic
	equations} or simply {\em filtering equations}. 
\begin{exercise} Suppose that $W$ and $B$ are not independent but
	$d[W, B]_t=\rho(X_t, Y_t)dt$ for some measurable function
	$\rho(x,y)$ which is bounded when $x$ belongs to a bounded interval. Obtain the filtering equations in this setting.
\end{exercise}
\begin{exercise} Extend the filtering equations to a multidimensional
	setting. 
\end{exercise}
\subsection{Kalman-Bucy filter}
Kalman-Bucy filter is a celebrated example of filtering which finds
widespread use in real-world problems. In particular, it will be essential for the solution of the equilibirum in Kyle model. 

We assume the signal process
satisfies an Ornstein-Uhlenbeck SDE:
\[
X_t=X_0 + B_t + \int_0^t a X_s ds,
\]
where $X_0$ is a normal random variable, and the observation process
is given by
\[
Y_t=W_t + \int_0^t c X_s ds.
\]
We assume $W$ and $B$ are independent so that $\rho \equiv 0$. Since
the bivariate process $(X,Y)$ is Gaussian, the conditional
distribution of $X$ given $Y$ is also Gaussian. The mean is given by
$\what{X}$ which is given by
\[
\what{X}_t= \bbE X_0 + \int_0^t a \what{X}_s ds + c\int_0^t
\left\{\what{X^2_s}-(\what{X}_s)^2\right\}dN_s \] by Theorem
\ref{filter}. Let $v_t:=\bbE[(X_t -\what{X}_t)^2|\cF^Y_t]$ be the
conditional variance of $X_t$ given $\cF^Y_t$. I.e. $v_t
=\what{X^2_t}-(\what{X}_t)^2$. Next, let's find the filtering
equation for $\what{X^2_t}$. Again, using It\^o's formula and
Theorem \ref{filter} \[
\what{X^2_t}= \bbE X^2_0
+ \int_0^t (1+2 a \what{X^2}_s) ds + c\int_0^t
\left\{\what{X^3_s}-\what{X^2_s}\what{X}_s\right\}dN_s. \] Recall
that for $Z \sim N(\mu, \sigma^2)$, $\bbE Z^3=\mu(\mu^2+3
\sigma^2)$. Thus, since $X_t$ is conditionally Gaussian,
\[
\what{X^3_s}-\what{X^2_s}\what{X}_s=\what{X}_s\left(\what{X^2_s} +3
v_s -\what{X^2_s}\right)= 2 v_s \what{X}_s.
\]
Thus, \bean
dv_t&=&d(\what{X^2_t}-\what{X}_t^2) \\
&=&2 c v_t \what{X}_t dN_t + (1+2 a \what{X^2_t})dt - 2
\what{X}_t(cv_tdN_t+ a \what{X}_tdt)-c^2v_t^2dt\\
&=&(1+ 2 a v_t - c^2 v_t^2)dt, \eean so that $v$ solves an ordinary
differential equation. This differential equation has a solution. If
$\beta>0$ and $-\gamma<0$ are two roots of the quadratic $1+2 a x -
c^2 x^2$, and if $\lambda=c^2(\beta+\gamma)$, then \bean
v_t&=&\frac{\delta \beta e^{\lambda t}-\gamma}{\delta e^{\lambda
		t}+1}, \qquad \mbox{where} \\
\delta&=&\frac{\sigma^2 + \gamma}{\beta-\sigma^2}, \eean and
$\sigma^2=\mbox{var}(X_0)$. Note that $v(\infty)=\beta$.
\begin{exercise} \label{f:ex:kyle}
	Let $X$ be the unobserved signal given by
	\[
	X_t= X_0 + \int_0^t \sigma(s)d W_s, \qquad \forall t \in [0,1],
	\]
	where $X_0$ is a normal random variable with mean zero and variance
	$s(0)$, and $\sigma$ is a continuous and deterministic function such
	that
	\[
	\int_0^1 \sigma^2(s)ds < \infty.
	\]
	There also exists an observation process $Y$ given by
	\[
	Y_t=B_t + \int_0^t \frac{X_s - Y_s}{f(s)}ds,
	\]
	where $f$ is a deterministic continuous function. Suppose $B$ and
	$W$ are independent Brownian motions.
	\begin{itemize}
		\item[a)] Find the  equation for $\hat{X}$ given the observation process
		$Y$. 
		\item[b)] Let $v(t):=\bbE[(X_t-\widehat{X}_t)^2|\cF^Y_t]$ where $\cF^Y$ is the filtration generated by $Y$ and $\widehat{X}_t=\bbE[X_t|\cF^Y_t]$. Show that $v$ solves the differential equation
		\[
		f^2(t) v'(t) + v^2(t)= \sigma^2(t)f^2(t).
		\]
		(Hint: If $Z \sim N(\mu, \sigma^2)$ then $\bbE Z^3=\mu(\mu^2+3
		\sigma^2)$.)
		\item[c)] Let $s(t):=s(0) + \int_0^t \sigma^2(s)ds$. Suppose
		$f(t)=s(t)-t$ and $s(t)-t\geq 0$. Using the differential equation above show that,
		given $\cF^Y_t$, $X_t$ is a normal random variable with mean
		$\widehat{X}_t$ and variance $s(t)-t$. Also show that  $\widehat{X}$ is a Brownian motion.
		\end{itemize}
\end{exercise}
Finally, Liptser and Shiryaev \cite{ls} is an excellent source for the fundamentals of stochastic filtering. It includes the results above and many more!
\section{Diffusion bridges} \label{s:db}
By a {\em diffusion bridge} we usually understand a conditioning of a given diffusion process to arrive at a fixed given value at some future time. The most well-known of such bridges is the {\em Brownian bridge}. More precisely, if $B$ is a standard Brownian motion and $x \in \bbR$,  we can construct a process $X^x$ such that the distribution of $X^x$ is that of $B$ conditioned on $B_1=x$. Note that since $P(B_1=x)=0$, the law of Brownian bridge is {\em not} absolutely continuous to that of $B$.  However, we shall see that
\[
\mbox{Law}(X^x_s; s\leq t) \sim \mbox{Law}(B_s; s\leq t), \qquad \forall t<1.
\]
One way to construct such a bridge is to define
\be \label{e:BB1st}
X^x_t:=B_t + (x-B_1)t.
\ee
The above constructs a Gaussian process  with  $E[X^x_t]=xt$, continuous on $[0.1]$, and  for $s<t$
\[
Cov(X_s^x,X_t^x)=E[(B_t-tB_1)(B_s-sB_1)]=s(1-t).
\]
To show that the above construction indeed is the desired Brownian bridge, let us verify that $B_t$ conditioned on $B_1=x$ has the above covariance structure and has the same mean. Indeed,
\be \label{e:BBlaw}
P(B_s\in dy, B_t\in dz|B_1=x )=\frac{p(s,x)p(t-s,z-y)p(1-t,x-z)}{p(1,x)}dydz,
\ee
where $p(\cdot,\cdot)$ is the transition density of $B$, i.e.
\[
p(s,y-z)dy=P(B_{t+s}\in dy|B_t=z)=\frac{1}{\sqrt{2\pi s}}\exp\left(-\frac{(y-z)^2}{2s}\right).
\]
\begin{exercise}
	Using (\ref{e:BBlaw}) show that $E[B_t|B_1=x]=tx$ and $Cov(B_s,B_t)=s(1-t)$ for $s<t<1$.
\end{exercise}
However, the above construction is not adapted to the filtration of $(B_t)$ and requires the knowledge of $B_1$. As such, this will not be useful to construct the bridge in the Kyle model since although the insider observes $B$ continuously in time she does not know the value of $B_1$ at time $t<1$. 

The second construction of a Brownian bridge uses SDEs. Consider
\be \label{e:BB2nd}
X^x_t= B_t+ \int_0^t \frac{x-X^x_s}{1-s}ds.
\ee
Its unique solution is given by
\be \label{e:BB2ndsol}
X^x_t=xt +(1-t)\int_0^t\frac{1}{1-s}dB_s
\ee
\begin{exercise}
	Verify the above using integration by parts.
\end{exercise}
It is clear from (\ref{e:BB2ndsol}) that $E[X^x_t]=xt$. Moreover,
\[
Cov(X^x_s,X^x_t)=(1-s)(1-t)E\left[\int_0^s\frac{1}{1-r}dB_r\int_0^t\frac{1}{1-r}dB_r\right]=(1-s)(1-t)\int_0^s\frac{1}{(1-r)^2}dr=s(1-t).
\]
It remains to show that $X^x$ is continuous on $[0,1]$. Continuity on $[0,1)$ is clear. What is left to show is that $X^x_t \rar x$ as $t \rar 1$.
\begin{exercise}
	Show that $P(\lim_{t\rar 1}X^x_t=x)=1$, where $X^x$ is defined by (\ref{e:BB2ndsol}) by observing that 
	\[
	\int_0^t\frac{1}{1-s}dB_s\eid W_{\frac{t}{1-t}},
	\]
	where $W$ is some Brownian motion. Use also the fact that $\lim_{t \rar \infty}\frac{W_t}{t}=0$ with probability $1$ for any Brownian motion.
\end{exercise}
\begin{exercise}
	Show that the solution of (\ref{e:BB2nd}) is a semimartingale. (Hint: Compute $E[|x-X^x_t]$.)
\end{exercise}
We have performed two different constructions of a Brownian bridge. Are there indeed very different or is there a link between them other than that they have the same distribution. 

To answer this question let us recall that the first construction is not adapted and depends on the knowledge of $B_1$. Let us now see what happens to $B$ if we enlarge its filtration with $B_1$. Let $(\cF_t)_{t \rar 1}$ be the natural filtration of $B$ and $\cG_t$ denote $\cF_t \vee \sigma(B_1)$ following the approach of Mansuy and Yor \cite{my}. To compute the decomposition of $B$ under $\cG_t$, take a text function $g$ and a set $A \in \cF_s$ and consider
\[
M_t:=E[g(B_1)|\cF_t]=\int g(y)p(1-t,y-B_t)dy.
\]
If we apply the Ito formula to $p$ and note that $p_t =\frac{1}{2}p_{xx}$, we obtain
\[
M_t=M_s -\int_s^t \int g(y)p_x(1-t,y-B_r)dydr.
\]
Thus,
\bean
E[(B_t-B_s)g(B_1)\chf_A]&=&E[(B_t-B_s)M_1\chf_A]\\
&=&E[(B_t-B_s)M_t\chf_A]\\
&=&E\left[([M,B]_t-[M,B]_s)\chf_A\right]\\
&=&-E\left[\chf_A\int_s^t\int g(y)p_x(1-t,y-B_r)dydr\right]\\
&=&-E\left[\chf_A\int_s^t\int g(y)\frac{p_x(1-t,y-B_r)}{p(1-t,y-B_r)}p(1-t,y-B_r)dydr\right]\\
&=&-E\left[\chf_A\int_s^t g(B_1)\frac{p_x(1-t,B_1-B_r)}{p(1-t,B_1-B_r)}dr\right]\\
&=&-E\left[\chf_A g(B_1)\int_s^t \frac{p_x(1-t,B_1-B_r)}{p(1-t,B_1-B_r)}dr\right].
\eean
However, the above implies that
\[
\beta_t:=B_t +\int_0^t \frac{p_x(1-t,B_1-B_r)}{p(1-t,B_1-B_r)}dr
\]
is a $\cG_t$-martingale. In view of L\`evy's characterisation of Brownian motion, $\beta$ must be a $\cG_t$-Brownian motion. Moreover, 
\[
\frac{p_x(t,y)}{p(t,y)}=-\frac{y}{t}.
\]Thus,
\[
B_t=\beta_t +\int_0^t \frac{B_1-B_r}{1-r}dr,
\]
where $\beta$ is a $\cG$-Brownian motion. Now, if we plug this into (\ref{e:BB1st}), we obtain
\bean
X^x_t&=& \beta_t +\int_0^t \left\{x-B_1+\frac{B_1-B_r}{1-r}\right\}dr\\
&=&\beta_t +\int_0^t\frac{x-xr + rB_1 -B_r}{1-r}dr=\beta_t +\int_0^t\frac{x-(B_r + r(x-B_1))}{1-r}dr\\
&=&\beta_t +\int_0^t\frac{x-X_r^x}{1-r}dr,
\eean
which is a weak solution of (\ref{e:BB2nd}).

\subsection{Absolute continuity of laws and bridges of general diffusions} Clearly, the probability law of Brownian bridge is not absolutely continuous with respect to that of Brownian motion when we consider the whole trajectory from $0$ to $1$. However, this does not imply the singularity of laws when the trajectories are confined to intervals of the form $[0,T]$ for $T<1$. 

Indeed, if we consider (\ref{e:BBlaw}) and compare it with the corresponding law for Brownian motion, we may conjecture that the law $P^{0\rar x}_{0 \rar 1}$ induced by the Brownian bridge on the space of continuous functions admits
\[
\frac{dP^{0\rar x}_{0 \rar 1}}{dP^0}\big|_{\cF_T}=p(1-T, x-X_T),\quad T<1,
\]
where $P^0$ is the law induced by the Brownian motion starting at $0$ and $X$ is the coordinate process. 

This guess can be verified by Girsanov transformation. To this end observe that $M_t:=p(1-t,x-B_t)$ is a bounded martingale on $[0,T]$ for $T<1$. Thus, if define a probability measure $Q$ on $\cF_T$ by
\[
\frac{dQ}{dP}=M_T,
\]
Girsanov theorem implies that $B$ follows under $Q$ the following dynamics:
\[
B_t= \beta_t +\int_0^t\frac{x-B_s}{1-s}ds
\]
for some $Q$-Brownian motion $\beta$.
That is, $B$ under $Q$ is a weak solution of (\ref{e:BB2nd}). Since the strong uniqueness holds on $[0,T]$, so does weak uniqueness. This in particular implies the aforementioned absolute continuity. 

The above procedure also hints us how to proceed in order to construct a bridge of a given diffusion. To make the construction precise suppose that $X$ is the unique solution of 
\[
X_t= y + \int_0^t \sigma(X_s)dB_s +\int_0^t b(X_s)ds,
\]
for some functions $\sigma$ and $b$. Suppose that the coefficients of the SDE are reqular enough so that the solution admits a transition density with respect to Lebesgue measure. Let $p(t,y,z)$ denote this density. That is, 
\[
p(t,y,z)dz=P(X_{s+t}\in dz|X_s=y).
\]
If we impose more conditions on the SDE, we can ensure that for any $z$ the mapping $(t,y)\mapsto p(t,y,z)$ is smooth enough to apply Ito's formula.  Let's suppose this is the case so that
\[
M_t:=p(1-t,X_t,x)=M_0+\int_0^t p_y(1-s,X_s,x)\sigma(X_s)dB_s
\]
is a martingale on $[0,T]$ for any $T<1$. Thus, Girsanov theorem yields a $Q_T$ on $\cF$ defined by
\[
\frac{dQ_T}{dP}=\frac{M_T}{M_0}
\]
under which $X$ follows
\[
X_t=y + \int_0^t \sigma(X_s)d\beta_s +\int_0^t \left\{b(X_s)+\sigma^2(X_s)\frac{p_y(1-s,X_s,x)}{p(1-s,X_s,x)}\right\}ds, \quad t<T. \]
It can be fairly easily shown that $Q_T$ converges to some measure $Q$, which we would like to identify with the law of $X$ conditioned to be equal to $x$ at time $1$. This requires certain measure theoretic technicalities and is beyond the scope of these notes. However, under fairly mild conditions, this can be done (see \cite{smb}) and the above recipe for the SDE for $X$ conditioned to arrive (continuously) at $x$ at time $1$ therefore works. 

A well-known example of the above construction is the $3$-dimensional Bessel bridge. That is, for $x \geq 0$,
\[
X_t =x+ \int_0^t\frac{1}{X_s}ds.
\]
The solution of the above SDE never hits $0$ after its initiation and converges to $\infty$ as $t \rar \infty$. But we can condition $X$ so that it converges to $0$ at time $1$ using the above recipe. There is a slight difficulty here since $p(1-t, y,0)=0$ (see Section XI.1 of Revuz and Yor \cite{RY} for the exact description of this density). To circumvent this, define
\[
M_t=h(t,X_t),
\]
where 
\[
h(t,y)= \lim_{z\rar 0}\frac{p(1-t,y,z)}{p(1,x,z)}.
\] In this case,
\[
\lim_{x \rar 0}\frac{h_y(t,y)}{h(t,y)}=-\frac{y}{1-t}.
\]
Thus, the corresponding SDE for the $3$-dimensional Bessel bridge from $x$ to $0$ is given by
\be \label{db:BesBr}
X_t=x + B_t +\int_{0}^t\left\{\frac{1}{X_s}-\frac{X_s}{1-s}\right\}ds.
\ee
\begin{remark}
	A similar connection with the enlargement of filtration theory exists in the case of general diffusions as well. Under certain integrability conditions and by repeating the same argument we used for Brownian motion we can show that \[
	dX_t=\sigma(X_t)d\beta_t + \left\{b(X_t)+\sigma^2(X_t)\frac{p_y(1-t,X_t,X_1)}{p(1-t,X_t,X_1)}\right\}dt,
	\]
	where $\beta$ is a $\cG$-Brownian motion while $\cG_t=\cF^X_t\vee \sigma(X_1)$.
\end{remark}
\section{Proof of Theorem \ref{mm:t:AC}}
	Given that we have the required machinery we can now return to the Kyle model and finds its equilibrium. Our first task is to prove Theorem \ref{mm:t:AC}.
	
	Using Ito's formula for general semimartingales (see, e.g. Theorem II.32 in \cite{Pro}) we obtain
	\bean
	dH(t, Y_t)&=&H_t(t,Y_{t-})dt + H_y(t,Y_{t-})dY_t + \half H_{yy}(t,Y_{t-})d[Y,Y]^c_t \\
	&&+ \left\{H(t,Y_t)-H(t,Y_{t-})-H_y(t,Y_{t-})\Delta Y_t \right\}\\
	&=& H_x(t,Y_{t-})w(t,Y_{t-})dY^c_t +dFV_t,
	\eean
	where $FV$ is of finite variation. Therefore,
	\be
	[\theta, S]^c_t =\int_0^t  H_y(s,Y_{s-})d[Y^c,\theta]_s =\int_0^t  H_y(s,Y_{s-})\left\{d[B,\theta]_s+d[\theta,\theta]^c_s\right\}. \label{mm:eq:THQV}
	\ee
	Moreover, integrating (\ref{mm:eq:insW}) by parts (see Corollary 2 of Theorem  II.22 in \cite{Pro}) we get
	\be \label{mm:eq:Wibp}
	W^{\theta}_1=f(Z_1)\theta_{1-}-\int_0^{1-}H(t,Y_{t-}))d\theta_t -[\theta,H(\cdot, Y)]_{1-}
	\ee
	since the jumps of $\theta$ are summable. 
	
	Consider the function
	\begin{equation}\label{mm:e:generalG_a}
	\Psi^a(t,x):=\int_{\xi(t,a)} ^x (H(t,u)-a)du+\frac{1}{2}\int_t^1H_y(s,\xi(s,a))ds
	\end{equation}
	where $\xi(t,a)$ is the unique solution of $H(t,\xi(t,a))=a$.
	Direct differentiation with respect to $x$ gives that
	\begin{equation}\label{mm:e:G_a_x}
	\Psi^a_x(t,x)=H(t,x)-a.
	\end{equation}
	Differentiating above with respect to $x$ gives
	\begin{equation}\label{mm:e:G_a_xx}
	\Psi^a_{xx}(t,x)=H_x(t,x).
	\end{equation}
	Direct differentiation of $\Psi^a(t,x)$ with respect to $t$ gives
	\begin{eqnarray}
	\Psi^a_t(t,x) &=&\int_{\xi(t,a)} ^x \frac{H_t(t,u)}{w(t,u)}du-\frac{1}{2}H_x(t,\xi(t,a))w(t,\xi(t,a)) \nn \\ 
	&=&-\frac{1}{2}H_x(t,x).\nn
	\end{eqnarray}
	 Combining the above with (\ref{mm:e:G_a_xx})
	 gives
	\begin{equation} \label{mm:eq:pdepsi}
	\Psi^a_t+\frac{1}{2}w(t,x)^2\Psi^a_{xx}=0.
	\end{equation}
	Applying Ito's formula we  deduce 
	\bean
	d\Psi^a(t,Y_t)&=&\Psi_t(t, Y_{t-})dt + (H(t,Y_{t-})-a)dY^c_t\nn \\
	&&+\half \Psi_{xx}(t,Y_{t-})d[Y,Y]^c_t + \Psi^{a}(t,Y_t)-\Psi^a(t,Y_{t-})\\
	&=&\left(H(t,Y_{t-})-a\right)dY^c_t+\half \Psi_{xx}(t,Y_{t-})\left(d[Y,Y]^c_t-dt\right)\nn \\
	&& + \Psi^{a}(t,Y_t)-\Psi^a(t,Y_{t-}) \nn \\
	&=&\left(H(t,Y_{t-})-a\right)dY^c_t+ \half H_y(t,Y_{t-})(d[Y,Y]^c_t-dt)\nn \\
	&& + \Psi^{a}(t,Y_t)-\Psi^a(t,Y_{t-})
	\eean
	where one to the last equality follows from (\ref{mm:eq:pdepsi}) and the last one is due to (\ref{mm:e:G_a_xx}). 
	
	The above implies
	\bean
	\Psi^a(1-,Y_{1-})&=&\Psi^a(0,0)+ \int_0^{1-}H(t,Y_{t-})(dB_t +d\theta_t) -a (B_1 +\theta_{1-})\\
	&&+ \half \int_0^{1-}H_y(t,Y_{t-})(d[Y,Y]^c_t-dt) \\
	&&+ \sum_{0<t<1}\left\{\Psi^{a}(t,Y_t)-\Psi^a(t,Y_{t-})-(H(t,Y_{t-})-a)\Delta\theta_t\right\}
	\eean
	Combining the above and (\ref{mm:eq:Wibp}) yields
	\bean
	E^{0,z}\left[W_1^{\theta}\right]&=&E^{0,z}\left[\Psi^{f(Z_1)}(0,0)-\Psi^{f(Z_1)}(1-,X_{1-})-f(Z_1)B_1+\int_0^{1-}H(t,X_{t-})dB_t\right. \nn \\
	&&+\half \int_0^{1-}w(t,X_{t-})H_x(t,X_{t-})(2d[B,\theta]_t+d[\theta,\theta]^c_t)\nn \\
	&&+ \sum_{0<t<1}\left\{\Psi^{f(Z)}(t,Y_t)-\Psi^{f(Z)}(t,Y_{t-})-(H(t,Y_{t-})-f(Z))\Delta\theta_t\right\}\nn\\
	&&-\left.\int_0^1  H_y(s,Y_{s-})w(s,Y_{s-})\left\{d[B,\theta]_s+d[\theta,\theta]^c_s\right\} - \sum_{0<t<1}(H(t,Y_t)-H(t,Y_{t-}))\Delta \theta_t \right]\nn \\
	&=&E^{0,z}\left[\Psi^{f(Z)}(0,0)-\Psi^{f(Z)}(1-,Y_{1-})-\half \int_0^{1-}H_y(t,Y_{t-})d[\theta,\theta]^c_t\right.\nn \\
	&&+\left. \sum_{0<t<1}\left\{\Psi^{f(Z_1)}(t,Y_t)-\Psi^{f(Z)}(t,Y_{t-})-(H(t,Y_{t})-f(Z))\Delta\theta_t\right\} \right]\nn \\
	&\leq&E^{0,z}\left[\Psi^{f(Z)}(0,0)-\Psi^{f(Z)}(1-,Y_{1-})\right]
	\eean
	since $H$ is increasing, and
	\bean
	\Psi^{a}(t,Y_t)-\Psi^{a}(t,Y_{t-})-(H(t,Y_t)-a)\Delta\theta_t&=&\int_{Y_{t-}}^{Y_t}(H(t,u)-a)du -(H(t,Y_t)-a)\Delta\theta_t \\
	&\leq&(H(t,Y_t)-a)\Delta Y_t  -(H(t,Y_t)-a)\Delta\theta_t=0. 
	\eean
	Note the inequality above becomes equality if and only if $\Delta \theta_t=0$ due to the strict monotonicity of $H$. Moreover, $\Psi^{f(Z)}(1-,Y_{1-})\geq 0$ with an equality if and only if $H(1-, Y_{1-})=f(Z)$. Therefore, $E^{0,z}\left[W_1^{\theta}\right]\leq E^{0,z}\left[\Psi^{f(Z)}(0,0)\right]$ for all admissible $\theta$s end equality is reached if and only if the following two conditions are met. 
	\begin{itemize}
		\item[i)] $\theta$ is continuous and of finite variation.
		\item[ii)] $H(1-,Y_{1-})=f(Z), \, P^{0,z}$-a.s..
	\end{itemize} 
	
	Hence, the proof will be complete if one can find a sequence of  absolutely continuous admissible strategies, $(\theta^n)_{n \geq 1}$ such that $\lim_{n \rar \infty}E^{0,z}\left[W_1^{\theta^n}\right]=E^{0,z}\left[\Psi^{f(Z_1)}(0,0)\right]$. However, note that the admissibility will be immediate as soon as $\theta$ is a semimartingale since $H$ is assumed to be bounded.
	
	Define
	\[
	Y_t:=B_t +\int_0^t \frac{H^{-1}(1,f(z))-Y_s}{1-s}ds.
	\]
	Recall that the above SDE has a semimartingale solution, which is a Brownian bridge converging a.s. to $H^{-1}(1,f(z))$.
	
	Set 
	\[
	d\theta_t=\frac{H^{-1}(f(z))-Y_t}{1-t}dt
	\]
	and observe that since $\theta$ is absolutely continuous, we have
	\[
	E^{0,z}[W^{\theta}]=E^{0,z}\left[\Psi^{f(Z)}(0,0)-\Psi^{f(Z)}(1,Y_{1}))\right].
	\]
	On the other hand, $\Psi^{f(Z)}(1,Y_{1}))=0$ since $H(1,Y_1)=f(Z)$. Thus, $\theta$ is optimal.
\section{Equilibrium}
Let us establish the equilibrium in the case of bounded asset value. 
\begin{theorem}
	\label{t:equilibrium} Suppose $f$ is bounded. Define $\theta$ by setting $\theta_0=$ and\[
	d\theta_t=\frac{Z-Y_t}{1-t}dt.
	\]
	Let $H$ be the unique solution of
	\[
	H_t +\half H_{yy}=0, \quad H(1,y)=f(y).
	\]
	Then, $(H,\theta)$ is an equilibrium.
\end{theorem}
\begin{proof}
	First note that since $f$ is bounded, $H$ is bounded by the same constant due to its Feynman-Kac representation.
	Thus, to show that $\theta$ is admissible it suffices to show that it is a semimartingale. Indeed, given $Z=z$
		\[
	Y_t=B_t +\theta
	\]
	is a Brownian bridge converging to $z$. Thus, $Y$ is a $P^{0,z}$-semimartingale for each $z$. Consequently, $\theta$ is a  $P^{0,z}$-semimartingale for each $z$. Moreover, $H(1,Y_1)=f(Z), \, P^{0,z}$-a.s.. Thus, $\theta$ is optimal given $H$. 
	
	Therefore, it remains to show that $H$ is a rational pricing rule. Note that if $Y$ is a Brownian motion in its own filtration, 
	\[
	H(t,y)=\bbE[f(Y_1)|\cF^Y_t]
	\]
	due to the Feynman-Kac representation of $H$, which in turn implies $H$ is a rational pricing rule. 
	
	Let us next show that $Y$ is a Brownian motion in its own filtration. This requires finding the conditional distribution of $Z$ given $Y$. Note that we are in fact in the setting of Exercise \ref{f:ex:kyle} with $\sigma(s)=0$ and $s(0)=1$.
	
	Exercise  \ref{f:ex:kyle} shows that the conditional distribution of $Z$ given $\cF^Y_t$ is Gaussian with mean $\what{X}_t:=\bbE[Z|\cF^Y_t]$ and variance $v(t)$, where 
	\[
	(1-t)^2 v'(t)+v^2(t)=0,
	\]
	and
	\[
	\what{X}_t= \int_0^t \frac{v(s)}{1-s}dN_s,
	\]
	where $N$ is the innovation process. 
	
	The unique solution of the  ODE with $v(0)=1$ is given by $v(t)=1-t$. Consequently, $\what{X}=N$, i.e. $\what{X}$ is an $\cF^y$-Brownian motion. Let us now see that $\what{X}=Y$. 
	
	Indeed,
	\[
	d\what{X}_t=dN_t= dY_t - \frac{\what{X}_t-Y_t}{1-t}dt.
	\]
	In other words, $\what{X}$ solves an SDE given $Y$. Since this is a linear SDE, it has a unique solution, which is given by $Y$ itself. Hence $Y$ is an $\cF^Y$-Brownian motion.	\end{proof}
Some remarks are in order. The {\em Kyle's lambda} or the market impact of trades is given by
\[
\lambda(t,y):=H_y(t,y).
\]
Thus, the flatter is $f$ the more liquid is the market. In other words, the market's liquidity depends crucially on the sensitivity of insider's valuation to changes in $Z$ (as soon as $Z$ is fixed to be a standard Normal at the beginning). 

Also note that the insider is indifferent among all bridge strategies that bring the market price to $f(Z)$ at time $1$. One such bridge is when 
\[
dY_t=dB_t+k\frac{Z-Y_t}{1-t}dt
\] 
for some $k$ while $H$ is still as in Theorem \ref{t:equilibrium}. Although this is optimal for the insider, it cannot make an equilibrium when combined with $H$ since $H(t,Y_t)$ will not be a martingale when $Y$ is as above.
\section{Dynamic private signal}
So far we have considered the case when the informed trader has a static signal giving her an unbiased predictor for the future asset value. However, most often the informed traders are investment banks with research divisions. For such traders it is not realistic to assume that they do not update their research in time. 

As a simple extension of the previous model suppose the informed trader receives a Gaussian signal $Z$ such that
\[
Z_t =Z_0 +\sigma B^Z_t,
\]
where $\sigma<1$ and $Z_0\sim N(0, 1-\sigma^2)$. Thus, $Z_1$ is a standard Normal random variable. We will again assume a strictly increasing $f$ such that $f(Z_1)$ will be the informed trader's unbiased estimator at time $1$  for the value $V$ of the asset that will be revealed at time $1$. The objective of the insider will be the same as (\ref{ins_obj}). 

If we go through Theorem \ref{mm:t:AC}, we realise that its proof does not depend on whether the insider's signal is static or dynamic. So the same proof can be used to show that if an admissible absolutely continuous trading strategy brings the market price to $f(Z_1)$, it will be optimal. Thus, $(H, \theta)$ will be an equilibrium if $H$ is a rational pricing rule and $\theta$ is such a bridge strategy.  

To achieve this, the insider need to  choose a $\theta$ such that $Y$ is a Brownian motion in its own filtration such that $Y=Z_1$.  Thus, if $H$ is as in Theorem \ref{t:equilibrium}, $(H,\theta)$ will be an equilibrium. 
\begin{theorem}
Let $s(t):=(1-\sigma^2)(1-t)$ and consider\[
Y_t= B_t + \int_{0}^t\frac{Z_r-Y_r}{s(r)}dr.
\]
Then, conditional distribution of $Z_t$ given $\cF^Y_t$ is Gaussian with mean $Y_t$ and variance $s(t)$. In particular, $Y$ is a Brownian motion in its own filtration.
\end{theorem}
\begin{proof}
	The proof is similar to the proof of Theorem \ref{t:equilibrium} and follows from Exercise \ref{f:ex:kyle}.
\end{proof}
Thus, if $\theta$ is chosen so that
\[
d\theta_t= \frac{Z_t-Y_t}{s(t)-t}dt
\]
and $H$ is as in Theorem \ref{t:equilibrium}, the pair $(H, \theta)$ is an equilibrium. The proof of that $\theta$ is a semimartingale is similar to the static bridge case and is left as an exercise.
\begin{remark}
	Observe that the market maker's pricing rules are the same in both static and dynamic private signal cases. Thus, one cannot deduce the type of private information by looking at market prices.
	
	When the private information is dynamic, the informed trader has a signal converging to the value of $Z_1$. Similarly, we can consider the total order process $Y$ as the signal of  market makers, again converging to $Z_1$. The remaining uncertainty at time $t$ for the informed trader is $Z_1-Z_t$ while it is $Z_1-Y_t=Y_1-Y_t$ for the market makers. The variance of $Z_1-Z_t$ is $s(t)$ and the variance of $Y_1-Y_t=1-t$. Since $s(t)<1-t$, the remaining uncertainty is lower for the informed trader, hence she has a competitive advantage.  
\end{remark}
\section{Risk averse market makers}
When the market makers are risk neutral as above, the equilibrium price evolves as a martingale in market makers' filtration. In \cite{RMM} the market makers are allowed to be risk-averse. More precisely, there are $N$ market makers with exponential utility aith the same risk aversion parameter $\rho$. The Bertran competition is interpreted as each market maker gaining or losing any expected utility in the equilibrium and the total order is split equally among them. That is, each market maker holds $-\frac{Y_t}{N}$ many shares and their utility evolves as a martingale.

Suppose that $d\theta_t =\alpha_tdt$ is an admissible trading strategy of the insider so that $Y$ in its own filtration satisfies
\[
dY_t= \sigma dB^Y_t +\hat{\alpha}_t dt,
\]
where $B^Y$ is an $\cF^M$-Brownian motion and $\hat{\alpha}$ is the $\cF^M$-optional projection of $\alpha$.  The best response of the market makers is to choose a price, $S$, that will satisfy the zero-utility gain condition. Let  price $S$ follow
\[
dS_t=Z_t dB^Y_t + \mu_t dt,
\]
for some predictable process $Z$ and an optional process $\mu$ that  are to be determined by the market makers.  As there is a potential discrepancy between $S_1$ and $V$, there is  a possibility of a jump in the market makers' wealth at time $1$. More precisely, 
\[
\Delta G_1= \frac{Y_1}{N}(S_1-V).
\]
However, the zero-utility gain condition implies
\[
1=\bbE\left[\exp\left(-\frac{\rho Y_1}{N}(S_1-V)\right)\bigg|\cF^M_1\right],
\]
which is equivalent to 
\be \label{e:bsdet}
\bbE\left[\exp\left(\frac{\rho Y_1}{N}V\right)\bigg|\cF^M_1\right]=\exp\left(\frac{\rho Y_1}{N}S_1\right).
\ee
On the other hand, if we compute the dynamics of $U(G)$ for $t<1$ by Ito's formula, we obtain
\[
dU(G_t)=U(G_t)\frac{\rho}{N}Y_t\left\{ \sigma_t dB^Y_t +\left(\mu_t +\frac{\rho}{2N} Y_t \sigma_t^2\right)dt\right\}.
\]
Reiterating the zero-utility gain condition for $t<1$ shows that we must have
\[
\mu_t=-\frac{\rho}{2 N}Y_t Z_t^2.
\]
Therefore, the zero-utility gain condition stipulates that  the price $S$ follows
\be \label{e:bsde}
dS_t =Z_t dB^Y_t -\frac{\rho}{2 N}Y_t Z_t^2 dt,
\ee
and the market makers' problem is to find $(Z, S)$ to solve (\ref{e:bsde}) with the terminal condition (\ref{e:bsdet}) given the total demand process $Y$.

The BSDE in (\ref{e:bsde}) is reminiscent of the quadratic BSDEs, which have been studied extensively, and the connection of which to problems arising in mathematical finance is well-established. The essential deviation  of (\ref{e:bsde}) from  the BSDEs considered in the mathematical finance literature   is that the coefficient of $Z_t^2$ in (\ref{e:bsde}) is $\frac{\rho}{2 N}Y_t$, which is in general unbounded. This makes the direct application of the results contained in the current literature for quadratic BSDEs to (\ref{e:bsde}) impossible. Moreover, most importantly, the boundary value of the above BSDE is highly non-standard and depends on the solution. 

However, if we turn to a Markovian equilibrium,  i.e. consider $S_t=H(t,Y_t)$, it is natural to  expect that in equilibrium $\hat{\alpha}_t=\hat{\alpha}_t(t,Y_t, S_t, Z_t)$ for some deterministic function $\hat{\alpha}$ so that
\be \label{e:bsdef}
dY_t= \sigma dB^Y_t +\hat{\alpha}(t,Y_t, S_t, Z_t) dt.
\ee 
Thus, if a Markovian equilibrium can be attained it will provide a Markovian solution to the FBSDE defined by (\ref{e:bsdet})-(\ref{e:bsdef}), where $\hat{\alpha}$ is the optimal drift chosen by the insider. 

We now turn to the optimisation problem for the insider when $S_t=H(t,Y_t)$ for an admissible pricing rule $H$. Repeating basically what we have done in earlier sections,  we see that $H$ must satisfy a backward heat equation
\[
H_t +\frac{1}{2}H_{yy}=0,
\]
and therefore  Ito's formula will yield that $S$ should satisfy
$$
dS_t= H_y(t,Y_t)dY_t.
$$ 
Combining this with (\ref{e:bsde}) and (\ref{e:bsdef}) implies
\[
z\hat{\alpha}(t,y,s,z)=-\frac{\rho}{2 N}y z^2,
\]
i.e.
\be \label{e:ahatopt}
\hat{\alpha}(t,y,s,z)=-\frac{\rho}{2 N}y z
\ee
as soon as we note that $z= H_y (t,y)$ by the choice of $S$. 

Moreover, as before, the sole criterion of  optimality for the insider is that the strategy  fulfils the bridge condition $H(1, Y_1)=V$.  Thus,   if a Markovian equilibrium exists, 
\be \label{e:s3:Yineq}
dY_t= \sigma dB^Y_t -\frac{\sigma^2 \rho}{2 N}Y_t H_y(t,Y_t),
\ee
and  $H$ solves the backward heat  equation above and satisfies $H(1,Y_1)=V$. 

It is shown in \cite{RMM} that under suitable conditions an equilibrium exists and is characterised by  the following system of equations:
\begin{align}
H_t & + \frac{1}{2} H_{yy}=0 \label{eq:NCH} \\
dY_t &= d\beta_t -\frac{ \rho}{2 N}Y_t H_y(t,Y_t)dt \label{eq:NCY} \\
V& \eid H(1,Y_1), \label{eq:NCF}
\end{align}
with $Y_0=0$ where $\beta$ is a Brownian motion on some given probability space and $Y$ is understood to be a strong solution of the forward SDE. Note that the terminal condition of the PDE is given by the distribution of $Y_1$, which itself depends on $H$. Thus, the solution typically requires a fixed point argument. The proof is based on Schauder's fixed point theorem applied to a suitable class of measures.  This in turn gives a solution of the aforementioned forward-backward SDE. 
\begin{remark}
	It can be seen from the above system that the market makers' inventories, $-Y$, are mean-reverting in the equilibrium since $H_y>0$. This is more consistent with the empirical studies on market makers. 
\end{remark} 
\section{Notes and some related literature}]
Characterisation of the optimal strategy of the insider as an absolutely continuous process driving the market price to its fundamental value goes  back to Back \cite{B}. Theorem \ref{mm:t:AC} proves this characterisation in the more general setting of this chapter using techniques from \cite{Wu} and \cite{Danilova}.  

Dynamic private signal for the informed investor was first considered in continuous time in \cite{BP}. More recently,  Campi, \c{C}etin and Danilova \cite{GBP} developed a Dynamic Markov bridge theory to study equilibrium when the insider's signal is a general diffusion process. However, the existence of equilibrium in general - even in a Markovian framework- is till an open question. See the discussion in \cite{GBP} in this respect.

Default risk was first incorporated into the Kyle model in by Campi and \c{C}etin \cite{CC}. A generalisation of this framework with dynamic signal  was done in \cite{CCD}.

The case of risk averse insider with exponential utility was first studied by Holden and Subrahmanyam \cite{HS94} in discrete time. The authors characterised the equilibrium as the solution of a system of equations, which they were able to solve numerically. This model was brought to continuous time by Baruch \cite{baruch}, who limited the insider's strategies to absolutely continuous ones with speed of trading being an affine function of equilibrium price. 

Holden and Subrahmanyam \cite{HS92} allow  multiplicity of insiders having the same information trading in the market. They found via numerical analysis that the competition among the insiders lead to a faster revelation of their private information. In fact in the continuous time limit of their model the insiders would reveal their information immediately. This observation led to the study of the case when the insiders' private signals are not perfectly correlated. This was first done by Foster and Viswanathan \cite{FV} in discrete time and later extended to continuous time by Back, Cao and Willard \cite{bcw}. 

Although generalisation of the noise traders' demand process has attracted relatively less attention in the literature, there are nevertheless several works that address this issue.  Volatility of the noise traders fluctuating randomly over time was considered in the recent work of Colin-Dufresne and Fos \cite{CDF}. Biagini et al. \cite{BHMO} studied the Kyle's model when the noise demand follows a fractional Brownian motion. Corcuera et al. \cite{CFNO} in the context of noise demand process being a L\'evy process concluded that an equilibrium cannot exist when a jump component is present. Indeed \c{C}etin and Xing \cite{CX} have shown that the insider must follow a mixed strategy, i.e. apply an additional randomisation, in the equilibrium when the noise buy and sell orders follow Poisson processes. 

Empirical studies suggest that a model with risk averse  market makers is more realistic. However, there are only two papers that tackle this problem in the literature. Subrahmanyam \cite{sub} allowed the market makers to be risk averse in a one-period Kyle model. \c{C}etin and Danilova \cite{RMM} have shown that an equilibrium exists in the continuous time version of the Kyle model with risk averse market makers. Existence of an equilibrium with dynamic inside information and risk averse market makers is still an open question.

\end{document}